\documentclass[11pt]{article}

\pdfoutput=1

\usepackage{fullpage}

\usepackage[T1]{fontenc}

\usepackage{ifthen}

\usepackage[bf,font=small,aboveskip=2pt,belowskip=-1em]{caption}

\newcommand{\annref}[2][]{\ref{#2}#1} 

\usepackage{enumitem}
\setlist{itemsep=0pt,parsep=0pt}             

\usepackage{booktabs}
\usepackage[usenames,dvipsnames,svgnames,table]{xcolor}
\usepackage{colortbl}
\usepackage{multirow}

\usepackage[ruled,vlined,linesnumbered,noend]{algorithm2e}
\SetNlSty{textsf}{}{:}
\DontPrintSemicolon
\SetKwComment{tcp}{$\rhd$\, }{}

\usepackage{amsmath,amssymb}
\usepackage{ktmath}
\usepackage{mathtools}
\usepackage{mathabx}
\usepackage{xfrac}
\newcommand{\remove}[1]{}

\usepackage{txfonts}                            
\usepackage[scaled=0.9]{helvet}

\usepackage{url}
\usepackage{graphicx}
\usepackage{wrapfig}
\usepackage{balance}  
\usepackage{tikz}
\usepackage{alltt}
\usetikzlibrary{positioning}

\newcommand{\Work}{\textsf{Work}}

\newcommand{\procName}[1]{\mbox{\normalfont\texttt{#1}}}
\newcommand{\nul}[0]{\mbox{\normalfont\textbf{null}}}

\newcommand{\parent}{\mbox{\normalfont\texttt{parent}}}

\newcommand{\find}{\mbox{\normalfont\texttt{find}}}
\newcommand{\union}{\mbox{\normalfont\texttt{union}}}

\newcommand{\myparagraph}[1]{\smallskip\noindent\textbf{#1}}

\newcommand{\note}[2]{
    \ifthenelse{\equal{\showComments}{yes}}{\textcolor{#1}{#2}}{}
}

\title{Work-Efficient Parallel and Incremental Graph Connectivity}
\author{
Natcha Simsiri\thanks{College of Information and Computer Sciences, University of Massachusetts-Amherst, {\tt nsimsiri@umass.edu}}
\and
Kanat Tangwongsan\thanks{Computer Science Program, Mahidol University International College, {\tt kanat.tan@mahidol.edu}}
\and
Srikanta Tirthapura\thanks{Department of Electrical and Computer Engineering, Iowa State University, {\tt snt@iastate.edu}}
\and
Kun-Lung Wu\thanks{IBM T.J. Watson Research Center {\tt klwu@us.ibm.com}}
}

\usepackage{microtype}
\begin{document}

\maketitle

\begin{abstract}
On an evolving graph that is continuously updated by a high-velocity stream of edges, how can one efficiently maintain if two vertices are connected? This is the connectivity problem, a fundamental and widely studied problem on graphs. We present the first shared-memory parallel algorithm for incremental graph connectivity that is both provably work-efficient and has polylogarithmic parallel depth. We also present a simpler algorithm with slightly worse theoretical properties, but which is easier to implement and has good practical performance.  Our experiments show a throughput of hundreds of millions of edges per second on a $20$-core machine.
%
\end{abstract}



\section{Introduction}
\label{sec:intro}
Graph connectivity is a fundamental problem with a long history. On an
undirected graph, the basic connectivity question is: \emph{given two vertices,
is there a path between them?}

Our work is motivated by the need for high throughput real-time streaming graph
analytics.  Every minute, a staggering amount of high-velocity linked data is
being generated from social media interactions, the Internet of Things (IoT)
devices, among others---and timely insights from them are much sought after.
These data are usually cast as a stream of edges with the goal of maintaining
certain local and global properties on the accumulated data.  Modern stream
processing systems such as IBM Infosphere Streams~\cite{IBM-streams} and Apache
Spark~\cite{ZDL+13} rely on parallel processing of input streams to achieve high
throughput and real-time analytics. However, these systems only provide the
software infrastructure; scalable, parallel, and dynamic graph algorithms are
still needed to make use of the potential of these systems.

As a first step towards efficient parallel and dynamic graph algorithms, we
consider the parallel \emph{incremental} graph connectivity problem in a setting
where edges and queries arrive in bulk. Tackling the parallel incremental version of the 
problem, which allows only addition of edges to the graph, is an important stepping
stone towards the more general problem of (fully) {\em dynamic} connectivity
that allows both addition and deletion of edges.

There exist sequential algorithms for incremental graph connectivity, starting
from the popular union-find data structure~\cite{T75}; but these are, for the
most part, unable to take advantage of parallelism. There exist parallel
algorithms for graph connectivity~(e.g., \cite{SDB14,Gazit91}), but these are,
for the most part, not incremental. None of these meet the need for high
throughput dynamic graph processing.

In order to make effective use of parallelism in stream processing, systems such
as Apache Spark~\cite{ZDL+13} use a model of ``discretized streams'',
where the incoming high-volume stream
is divided into a sequence of ``minibatches''. Each minibatch is processed using a parallel
computation, and the resulting system can potentially achieve a very high throughput, subject
to the availability of appropriate algorithms. We adopt this model in our work and seek parallel 
methods that can process a minibatch of edges efficiently.


\myparagraph{Model:} On a vertex set $V$, a graph stream $\mathcal{A}$ is a sequence of minibatches $A_1, A_2, \dots$, where each minibatch $A_i$ is a set of edges on $V$. The graph at the end of observing $A_t$, denoted by $G_t$, is $G_t = (V, \cup_{i=1}^t A_i)$ containing all the edges up to $t$. The minibatches $A_i$ need not be of equal sizes.

In this paper, we study a \emph{bulk-parallel incremental connectivity problem}, 
which is to maintain a data structure that provides two operations:
\procName{Bulk-Update} and \procName{Bulk-Query}.  The \procName{Bulk-Update}
operation takes as input a minibatch of edges $A_i$ and adds them to the graph. The
\procName{Bulk-Query} operation takes a minibatch of vertex-pair queries and returns
for each query, whether the two vertices are connected on the edges observed so
far in the stream. On this data structure, the \procName{Bulk-Query} and
\procName{Bulk-Update} operations are each invoked with a (potentially large)
minibatch of input, each processed using a parallel computation. But a bulk
operation, say a \procName{Bulk-Update}, must complete before the next
operation, say a \procName{Bulk-Query}, can begin.

\myparagraph{Contributions:} We present the first shared-memory
parallel algorithm for incremental connectivity 
that is both provably work-efficient and has polylogarithmic parallel
depth.  We make the following specific contributions:
\begin{description}[itemsep=2pt,leftmargin=0pt,parsep=0pt,topsep=2pt]
\item[---]\emph{Simple Parallel Incremental Connectivity.}  We first present a
  simple algorithm that is easy to implement, yet has good theoretical
  properties. On a graph with $n$ vertices, this algorithm makes a single pass
  through the stream using $O(n)$ memory, and can process a minibatch of $b$
  edges, using $O(b \log n)$ work and $O(\polylog(n))$ parallel depth.  We
  describe this algorithm in Section~\ref{sec:simple-algo}.

\item[---]\emph{Work-Efficient Parallel Incremental Connectivity.} We present an
  improved parallel algorithm with total work $O((m+q)\alpha(m+q, n))$ where $m$
  is the total number of edges across all minibatches, $q$ is the total number
  of connectivity queries across all minibatches, and $\alpha$ is an inverse
  Ackermann's function (see Section~\ref{sec:related}). This matches the work of
  the best sequential counterpart, which makes this parallel algorithm {\em
    work-efficient}.  Further, the parallel depth of processing a minibatch is
  polylogarithmic. Hence, the sequential bottleneck in the runtime of the
  parallel algorithm is very small, and the algorithm is capable of using almost
  a linear number of processors efficiently. We are not aware of a prior
  parallel algorithm with such provable properties on work and depth. We
  describe this algorithm in Section~\ref{sec:work-efficient}.

\item[---]\emph{Implementation and Evaluation.} We implemented and benchmarked
  a variation of our simple parallel algorithm on a shared-memory machine. 
  Our experimental results show that the algorithm
  achieves good speedups in practice and is able to efficiently use the
  available parallelism.  On a 20-core machine, 
  it can process hundreds of millions of edges per second, and
  realize a speedup of $8$--$11$x over its single threaded performance. Further
  analysis shows good scalability properties as the number of threads is
  varied. We describe this in Section~\ref{sec:eval}.
\end{description}


\section{Related Work}
\label{sec:related}

Let $n$ be the number of vertices, $m$ the number of operations, and $\alpha$ an
inverse Ackermann's function (very slow-growing, practically a constant
independent of $n$).  In the sequential setting, the basic data structure for
incremental connectivity is the well-studied union-find data
structure~\cite{CormenLRS:book09}.  Tarjan~\cite{T75} achieves an
$O(\alpha(m, n))$ amortized time per \find{}, which has been shown to be optimal
(see Seidel and Sharir~\cite{SeidelS05} for an alternate analysis).

Recent work on streaming graph algorithms focuses on minimizing the memory
requirement, with little attention given to the use of parallelism. This line of
work has largely focused on the ``semi-streaming model''~\cite{FKMSZ05}, which
allows $O(n \cdot \polylog(n))$ space usage.  In this model, the union-find data
structure~\cite{T75} solves incremental connectivity in $O(n)$ space and a total
time nearly linear in $m$.

When only $o(n)$ of workspace (sublinear) is allowed, interesting tradeoffs are
known for multi-pass algorithms.  For an allotment of $O(s)$ workspace, an
algorithm needs $\Omega(n/s)$ passes~\cite{FKMSZ05} to compute the connected
components of a graph.  Demetrescu et al.~\cite{DFR09} consider the W-stream
model, which allows the processing of streams in multiple passes in a pipelined
manner: the output of the $i$-th pass is given as input to the $(i+1)$-th
pass. They show a tradeoff between the number of passes and the memory
required. With $s$ bits of space, their algorithm computes connected components
in $O((n\log n)/s)$ passes. Demetrescu et al.~\cite{DEMR10} present a simulation
of a PRAM algorithm on the W-Stream model, allowing existing PRAM algorithms to
run \emph{sequentially} in the W-Stream model.

McColl et al.~\cite{MGB13} present a parallel algorithm for maintaining
connected components in a fully dynamic graph, which handles edge deletions---a
more general setting than ours.  As part of a bigger project (STINGER), their
work focuses on engineering algorithms that work well on real-world graphs and
gives no theoretical analysis of the parallel complexity.  In contrast, this
work focuses on achieving the best theoretical efficiency, matching the work of
the best sequential counterpart.


Berry et al.~\cite{BO+2013} present methods for maintaining connected components
in their parallel graph stream model, called X-Stream, which periodically ages
out edges. Their algorithm is essentially an ``unrolling'' of the algorithm
of~\cite{DFR09}, and edges are passed from one processor to another until the
connected components are found by the last processor in the sequence.  Compared
to our work, the input model and notions of correctness differ.  Our work views
the input stream a sequence of batches, each a set of edges or a set of queries,
which are unordered within the set.  Their algorithm strictly respects the
sequential ordering the edges and queries.  Further, they age out edges (we do
not).  Also, they do not give provable parallel complexity bounds.

There are multiple parallel (batch) algorithms for graph connectivity
including~\cite{SDB14,Gazit91} that are work-efficient (linear in the number of
edges) and that have polylogarithmic depth. Prior work on wait-free
implementations of the union-find data structure~\cite{AW91} focuses on the
asynchronous model, where the goal is to be correct under all possible
interleavings of operations; unlike us, they do not focus on bulk processing of
edges. 
There is also a long line of work on sequential algorithms for maintaining graph
connectivity on an evolving graph. See the recent work by~\cite{KKM13} that
addresses this problem in the general dynamic case and the references therein.




\section{Preliminaries and Notation}
\label{sec:prelim}
Throughout the paper, let $[n]$ denote the set $\{0,1,\dots, n\}$. A sequence is
written as $X = \langle x_1, x_2, \dots, x_{|X|} \rangle$, where $|X|$ denotes
the length of the sequence. For a sequence $X$, the $i$-th element is denoted by
$X_i$ or $X[i]$.  Following the set-builder notation, we denote by
$\langle f(x) \,:\, \Phi(x) \rangle$ a sequence generated (logically) by taking
all elements that satisfy $\Phi(x)$, preserving their original ordering, and
transform them by applying $f$.  For example, if $T$ is a sequence of numbers,
the notation $\langle 1 + f(x) \,:\, x \in T \text{ and } x \text{ odd}\rangle$
means a sequence created by taking each element $x$ from $T$ that are odd and
map $x$ to $1 + f(x)$, retaining their original ordering.  Furthermore, we write
$S \oplus T$ to mean the concatenation of $S$ and $T$.

We design algorithms in the work-depth model assuming an underlying CRCW PRAM
machine model. As is standard, the \emph{work} of an algorithm is the total
operation count, and the \emph{depth} (also called parallel time or span) is the
length of the longest chain of dependencies within a parallel computation.  The
gold standard for algorithms in this model is to perform the same amount of work
as the best sequential counterpart (work efficient) and to have polylogarithmic
depth.  We remark that an algorithm designed for the CRCW model can work in
other shared memory models such as EREW PRAM, with a depth that is a logarithmic
factor worse. 


We use standard parallel operations such as filter, prefix sum, map (applying a
constant-cost function), and pack, all of which has $O(n)$ work and at most
$O(\log^2 (n))$ depth on an input sequence of length $n$.  Given a sequence of
$m$ numbers, there is a duplicate removal algorithm \procName{removeDup} running
in $O(m)$ work and $O(\log^2 m)$ depth~\cite{JaJa:book92}. We also use the
following results to sort integer keys in a small range faster than a typical
comparison-based algorithm:
\begin{theorem}[Parallel Integer Sort~\cite{RajasekaranR89}]
  \label{thm:intsort}
  There is an algorithm \procName{intSort} that takes a sequence of integer keys
  $a_1, a_2, \dots, a_n$, each a number between $0$ and $c\cdot{}n$, where $c =
  O(1)$, and produces a sorted sequence in $O(n)$ work and $\polylog(n)$ depth.
\end{theorem}

\myparagraph{Parallel Connectivity:} For a graph $G = (V, E)$, a connected
component algorithm (\procName{CC}) computes a sequence of connected components
of $G$ $\langle C_i \rangle_{i=1}^k$, where each $C_i$ is a list of vertices in
the component. There are algorithms for \procName{CC} that have $O(|V| + |E|)$
work and $O(\polylog(|V|, |E|))$ depth (e.g., \cite{Gazit91, SDB14}), with
Gazit's algorithm~\cite{Gazit91} requiring $O(\log |V|)$ depth. 






\section{Simple Bulk-Parallel Data Structure}
\label{sec:simple-algo}
This section describes a simple bulk-parallel data structure for incremental
graph connectivity. We describe theoretical improvements to this basic version
in the next section. As before, $n$ is the number of vertices in the graph
stream.  The main result for this section is as follows:

\begin{theorem}
\label{thm:infinite-inc}
There is a bulk-parallel data structure for incremental connectivity, given by
Algorithms \procName{Simple-Bulk-Query} and \procName{Simple-Bulk-Update}, where
\begin{enumerate}[label=(\arabic*),topsep=0pt,parsep=0pt]
\item The total memory consumption is $O(n)$ words.
\item A minibatch of $b$ edges is processed by \procName{Simple-Bulk-Update} in $O(\log (\min\{b,n\}))$ parallel depth and $O(b \log n)$ total work.
\item A minibatch of $q$ connectivity queries, each asking for connectivity between two vertices, is answered by \procName{Simple-Bulk-Query} in $O(\log n)$ parallel depth and $O(q\log n)$ total work.
\end{enumerate}
\end{theorem}

In a nutshell, we show how to bootstrap a standard union-find structure to take
advantage of parallelism while preserving the height of the union-find forest to
be at most $O(\log n)$. For concreteness, we will work with union by size,
though other variants (e.g., union by rank) will also work.

\myparagraph{Union-Find:} We review a basic union-find implementation that uses
union by size.  From the viewpoint of graph connectivity, union-find maintains
connectivity information about a graph with vertices $V = [n]$ supporting:
\begin{itemize}[leftmargin=1em]
\item for $u \in V$, $\find(u) \in V$ returns an identifier of the connected
  component that $u$ belongs to. This has the property that
  $\find(u) = \find(v)$ if and only if $u$ and $v$ are connected in the graph.
\item for $u,v \in V$, $\union(u, v)$ links $u$ and $v$ together, making them in
  the same connected component.  It also returns the identifier of the component
  that both $u$ and $v$ now belong to---this is the same identifier one would
  get from running $\find(u)$ or $\find(v)$ at this point.
\end{itemize}

Conceptually, this data structure maintains a union-find forest, one tree for
each connected component.  In this view, $\find(u)$ returns the vertex that is
the root of the tree containing $u$ and $\union(u, v)$ joins together the roots
of the tree containing $u$ and the tree containing $v$.  The trees in a
union-find forest are typically represented by remembering each node's parent,
in an array $\parent{}$ of length $n$, where $\parent[u]$ is the tree's parent
of $u$ or $\parent[u] = u$ if it is the root of its component.

The running time of the union and find operations depends on the maximum height
of a tree in the union-find forest.  To keep the height small, at most
$O(\log n)$, a simple strategy, known as \emph{union by size}, is for \union{}
to always link the tree with fewer vertices into the tree with more vertices.
The data structure also keeps an array for the sizes of the trees.  The
following results are standard (see~\cite{SeidelS05}, for example):
\begin{lemma}[Sequential Union-Find]
  \label{lemma:seq-uf-cost}
  On a graph with vertices $[n]$, a sequential union-find data structure
  implementing the union-by-size strategy consumes $O(n)$ space and has the
  following characteristics:
  \begin{itemize}[topsep=2pt,itemsep=0pt,parsep=0pt,leftmargin=1em]
  \item Every union-find tree has height $O(\log n)$ and each \find{} takes
    $O(\log n)$ sequential time.
  \item Given two distinct roots $u$ and $v$, the operation $\union(u, v)$
    implementing union by size takes $O(1)$ sequential time.
  \end{itemize}
\end{lemma}

Our data structure maintains an instance of this union-find data structure,
called $U$. Notice that the \find{} operation is read-only.  Unlike the more
sophisticated variants, this version of union-find does not perform path
compression.




\subsection{Answering Connectivity Queries in Parallel}
Connectivity queries can be easily answered in parallel, using read-only
$\find$s on $U$.  To answer whether $u$ and $v$ are connected, we compute
$U.\find(v)$ and $U.\find(u)$, and report if the results are equal.  To answer
multiple queries in parallel, we note that because the $\find$s are read-only,
we can answer all queries simultaneously independently of each other.  We
present \procName{Simple-Bulk-Query} in Algorithm~\ref{algo:basic-bulk-query}.
\begin{algorithm}
\caption{$\procName{Simple-Bulk-Query}(U, \langle (u_i,v_i)\rangle_{i=1}^q$).}
\label{algo:basic-bulk-query}
\small
\KwIn{$U$ is the union find structure, and $(u_i,v_i)$ is a pair of vertices, for $i=1,\ldots,q$.}
\KwOut{For each $i$, whether or not $u_i$ is connected to $v_i$ in the graph.}

\For(in parallel){$i=1,2,\dots, q$}
{
  $a_i \gets (U.\find(u_i) == U.\find(v_i))$\;
}
\Return $\langle a_1, a_2, \ldots, a_q \rangle$\;
\end{algorithm}

Correctness follows directly from the correctness of the base union-find
structure.  The parallel complexity is simply that of applying $q$ operations of
$U.\find{}$ in parallel:

\begin{lemma}
\label{lemma:wd-query}
The parallel depth of \procName{Simple-Bulk-Query} is $O(\log n)$, and the work
is $O(q \log n)$, where $q$ is the number of queries input to the algorithm.
\end{lemma}


\subsection{Adding a Minibatch of Edges}
\emph{How can one incorporate (in parallel) a minibatch of edges $A$ into an
  existing union-find structure?} Sequentially, this is simple: invoke $\union$
on the endpoints of every edge of $A$.  To make it parallel, though, we cannot
blindly apply the $\union$ operations in parallel. Because $\union$ updates the
forest, running multiple $\union$ operations independently in parallel can
create inconsistencies in the structure.

We observe, however, that it is safe run multiple $\union$s in parallel as long
as they operate on different trees. This is not sufficient, as there may be a
number of union operations involving the same tree, and running these
sequentially will result in a large parallel depth.  For instance, consider
adding the edges of a star graph (with a very high degree) to an empty graph.
Because all the edges share a common endpoint, the center of the star is
involved in every $\union$, and hence no two operations can proceed in parallel.


To tackle this problem, our algorithm transforms the minibatch of edges $A$ into
a structure that can be connected up easily in parallel.  For illustration, we
revisit the example when the minibatch is itself a star graph.  Suppose there
are seven edges within the minibatch:
$(v_1,v_2), (v_1,v_3),(v_1,v_4),\ldots, (v_1,v_8)$.  By examining the minibatch,
we find that all of $v_1, \dots, v_8$ will belong to the same component. We now
apply these connections to the graph.

In terms of connectivity, it does not matter whether we apply the actual edges
that arrived, or a different, but equivalent set of edges; it only matters that
the relevant vertices are connected up. To connect up these vertices, our
algorithm schedules the $\union$s in only three parallel rounds as follows. The
notation $X \Vert Y$ indicates that $X$ and $Y$ are run in parallel:
\begin{enumerate}[label={\arabic*:},topsep=2pt,parsep=1pt,itemsep=0pt,leftmargin=1.5em]
\item \small $\union(v_1,v_2) \Vert \union(v_3,v_4) \Vert \union(v_5,v_6) \Vert \union(v_7,v_8)$
\item \small $\union(v_1,v_3) \Vert \union(v_5,v_7)$
\item \small $\union(v_1,v_5)$
\end{enumerate}

As we will soon see, such a schedule can be constructed for a component of any
size provided that no two of vertices in the component are connected previously. The resulting
parallel depth is logarithmic in the size of the minibatch.

\begin{algorithm}[tbh]
\caption{$\procName{Simple-Bulk-Update}(U,A)$}
\label{algo:simple-bulk-update}
\small
\KwIn{$U$: the union find structure,  $A$: a set of edges to add to the graph.}

\tcp{Relabel each $(u,v)$ with the roots of $u$ and $v$}
$A' \gets \langle (p_u, p_v) \,:\, (u, v) \in A \text{ where } p_u = U.\find(u) \text{ and } p_v = U.\find(v) \rangle$\;
\tcp{Remove self-loops}
$A'' \gets \langle (u, v) \,:\, (u, v) \in A' \text{ where } u \neq v \rangle$\;

$\mathcal{C} \gets \procName{CC}(A'')$\;

\ForEach(in parallel){$C \in \mathcal{C}$}
{
  $\procName{Parallel-Join}(U, C)$
}
\end{algorithm}

To add a minibatch of edges, our \procName{Simple-Bulk-Update} algorithm,
presented in Algorithm~\ref{algo:simple-bulk-update}, proceeds in three steps:

\smallskip
\noindent$\rhd$\,\textbf{Step 1:} Relabel edges as links between existing
components.  An edge $\{u, v\} \in A$ does not simply join vertices $u$
and $v$.  Due to potential existing connections in $G$, it
joins together $C_u$ and $C_v$, the component containing $u$ and the component
containing $v$, respectively.  In our representation, the identifier of the
component containing $u$ is $U.\find(u)$, so $C_u = U.\find(u)$ and similarly
$C_v = U.\find(v)$.  Lines~1-2 in Algorithm~\ref{algo:simple-bulk-update} create
$A''$ by relabeling each endpoint of an edge with the identifier of its
component, and dropping edges that are within the same component.

\smallskip
\noindent$\rhd$\,\textbf{Step 2:} Discover new connections arising from $A$.
After the relabeling step, we are implicitly working with the graph
$\tilde{H} = (V_{\tilde{H}}, A'')$, where $V_{\tilde{H}}$ is the set of all
connected components of $G$ that pertain to $A$ (i.e., all the roots in the
union-find forest reachable from vertices incident on $A$) and $A''$ is the
connections between them.  In other words, $\tilde{H}$ is a graph on
``supernodes'' and the connections between them using the edges of $A$.  In this
view, a connected component on $\tilde{H}$ represents a group of existing
components of $G$ that have just become connected as a result of incorporating
$A$.  While never materializing the vertex set $V_{\tilde{H}}$, Line~3 in
Algorithm~\ref{algo:simple-bulk-update} computes $\mathcal{C}$, the set of
connected components of $\tilde{H}$, using a linear-work
parallel algorithm for connected components, \procName{CC} (see Section~\ref{sec:prelim}).

\smallskip
\noindent$\rhd$\,\textbf{Step 3:} Commit new connections to $U$.  With the
preparation done so far, the final step only has to make sure that the pieces of
each connected component in $\mathcal{C}$ are linked together in $U$.  Lines~4-5
of Algorithm~\ref{algo:simple-bulk-update} go over the components of
$\mathcal{C}$ in parallel, seeking help from \procName{Parallel-Join}, the real
workhorse that links together the pieces.

\myparagraph{Connecting a Set of Components within $U$:} Let $v_1, v_2, \dots, v_k \in [n]$
be distinct tree roots from the union-find forest $U$ that form a component in $\mathcal{C}$, and need
to be connected together. Algorithm \procName{Parallel-Join} connects them up
in $O(\log k)$ iterations using a divide-and-conquer approach.  Given a sequence of
tree roots, the algorithm splits the sequence in half and recursively connects
the roots in the first half, in parallel with connecting the roots in the
second half. Since components in the first half and the second half
have no common vertices, handling them in parallel will not cause a conflict.
Once both calls return with their respective new roots, they are
unioned together.

\begin{algorithm}
\caption{$\procName{Parallel-Join}(U, C)$}
\label{algo:simple-join}
\small
\KwIn{$U$: the union-find structure, $C$: a seq.~of tree roots}
\KwOut{The root of the tree after all of $C$ are connected}
\If{$|C|==1$}{\Return $C[1]$}
\Else{
  $\ell \gets \lfloor |C|/2 \rfloor$\;
  $u \gets \procName{Parallel-Join}(U, C[1,2,\ldots,\ell])$ \textbf{in parallel with }
  \quad $v \gets \procName{Parallel-Join}(U, C[\ell+1,\ell+2,\ldots,|C|])$\;
  \Return {$U.\union(u, v)$}\;
}
\end{algorithm}

Correctness of \procName{Parallel-Join} is immediate since the order that the \union{} calls are made does not matter, and we know that different \union{} calls that proceed in parallel always work on separate sets of tree roots, posing no conflicts. 

\begin{lemma}
\label{lemma:simple-join-complexity}
Given $k$ distinct roots of $U$, Algorithm \procName{Parallel-Join} runs in $O(k)$ work and $O(\log k)$ depth.
\end{lemma}


\begin{lemma}[Correctness of~\procName{Simple-Bulk-Update}]
If $U$ is the shared-memory union-find data structure formed by a sequence of minibatch arrivals whose union equals the graph $G$, then for any $u, v \in V$, $U.\find(u) = U.\find(v)$ if and only if $u$ and $v$ are connected in $G$.
\end{lemma}

\begin{proof}
Consider a minibatch of edges $A$. Let $G_1$ be the set of edges that arrived prior to $A$ and $U_1$ the state of the union-find structure formed by inserting $G_1$. Let $G_2 = G_1 \cup A$ and let $U_2$ be the state of the union-find structure after \procName{Simple-Bulk-Update}$(U_1,A)$.  We will assume inductively that $U_1$ is correct with respect to $G_1$ and show that $U_2$ is correct with respect to $G_2$.

Let $x \neq y$ be a pair of vertices in $V$. We consider the following two cases.


{\bf Case I: $x$ and $y$ are not connected in $G_2$.} In this case, $x$ and $y$ are not 
connected in $G_1$ either. Let $r_x = U_1.\find(x), r_y = U_1.\find(y)$. From the inductive assumption, we know $r_x \neq r_y$.
Note that $A$ will not contain a path between $x$ and $y$. Hence in ${\cal C}$, the connected
components of $A''$, $r_x$ and $r_y$ will not be in the same component. When ${\cal C}$ 
is applied to $U_1$ in \procName{Parallel-Join}, the components containing $r_x$ and $r_y$ 
are not linked together, and hence it is still true that $U_2.\find(x) = U_2.\find(r_x) \neq U_2.\find(r_y) = U_2.\find(y)$.

{\bf Case II: $x$ and $y$ are connected in $G_2$.}
There must be a path $x = v_1, v_2, \dots, v_t = y$ in $G_2$. We will show that $U_2.\find(v_1) = U_2.\find(v_2) = \ldots = U_2.\find(v_t) $, leading to the conclusion $U_2.\find(x) = U_2.\find(y)$. Consider any pair $v_i$ and $v_{i+1}$, $1 \le i \le (t-1)$. Let $r_i=U_1.\find(v_i)$ and $r_{i+1}=U_1.\find(v_{i+1})$ denote the roots of the trees that contain $v_i$ and $v_{i+1}$ respectively in $U_1$. Suppose that $r_i=r_{i+1}$, then it will remain true that $U_2.\find(v_i)=U_2.\find(r_i)=U_2.\find(r_{i+1})=U_2.\find(v_{i+1})$.
Next consider the case $r_i \neq r_{i+1}$. Then $v_i$ and $v_{i+1}$ are not connected in $G_1$. To see this, suppose that $v_i$ and $v_{i+1}$ were connected in $G_1$. Then, $U_1.\find(v_i) = U_1.\find(v_{i+1})$, and it will remain true that $r_i=U_2.\find(v_i)=U_2.\find(v_{i+1})=r_{i+1}$.
In Steps 1 and 2 of \procName{Simple-Bulk-Update}, the edge $(r_i,r_{i+1})$ is inserted into $A''$ (note this edge is not a self-loop and is not eliminated in Step 2). In Step 3, when the connected components of $A''$ are computed, $r_i$ and $r_{i+1}$ are in the same component of ${\cal C}$. In $\procName{Parallel-Join}$, the subtrees rooted at $r_i$ and $r_{i+1}$ are unioned into the same component in $U_2$. As a result, $U_2.\find(r_i)=U_2.\find(r_{i+1})$. Since $U_2.\find(v_i) = U_2.\find(r_i)$ and $U_2.\find(v_{i+1}) = U_2.\find(r_{i+1})$, we have $U_2.\find(v_i)=U_2.\find(v_{i+1})$. Proceeding thus, we have $U_2.\find(x)=U_2.\find(y)$ in Case II.
\end{proof}


\smallskip 
\begin{lemma}[Complexity of \procName{Simple-Bulk-Update}]
\label{lemma:wd-update}
Given a minibatch $A$ with $b$ edges, \procName{Simple-Bulk-Update} takes $O(b \log n)$
work and $O(\log n)$ depth.
\end{lemma}

\begin{proof}
  There are three parts to the work and depth of \procName{Simple-Bulk-Update}.
  First is the generation of $A'$ and $A''$.  For each $(u, v) \in A$, we invoke
  $U.\find{}$ on $u$ and $v$, requiring $O(\log n)$ work and depth per edge.
  Since the edges are processed in parallel, this leads to $O(b \log n)$ work
  and $O(\log n)$ depth.  Then, $A''$ is derived from $A'$ through a parallel
  filtering algorithm, using $O(|A'|) = O(b)$ work and $O(1)$ depth. The second
  part is the computation of connected components of $A''$ which can be done in
  $O(|A''|) = O(b)$ work and $O(\log n)$ depth using the algorithm of
  Gazit~\cite{Gazit91}.  The third part is \procName{Parallel-Join}.  As the
  number of components cannot exceed $b$, and using
  Lemma~\ref{lemma:simple-join-complexity}, we have that the total work in
  \procName{Parallel-Join} is $O(b)$ and depth is $O(\log n)$.  Adding the three
  parts, we arrive at the lemma.
\end{proof}


\section{Work-Efficient Parallel Algorithm}
\label{sec:work-efficient}

\newcommand\ncoord[2][0,0]{%
    \tikz[remember picture,overlay]{\path (#1) coordinate (#2);}%
}
\tikzstyle{mybox} = [draw=gray!75, fill=gray!15, thick,
    rectangle, rounded corners, inner sep=5pt]

\begin{algorithm*}[thb]
  \caption{$\procName{Bulk-Find}(U, S)$---find the root in $U$ for each
    $s \in S$ with path compression.}
  \label{algo:bulk-find}
  \small
  \KwIn{$U$ is the union find structure. For $i=1,\ldots,|S|$, $S[i]$ is a
    vertex in the graph} 
  \KwOut{A response array $\textit{res}$ of length $|S|$ where $\textit{res}[i]$
    is the root of the tree of the vertex $S[i]$ in the input.}

  \tcp{\textbf{Phase I:} Find the roots for all queries}
  $R_0 \gets \langle (S[k], \nul) \,:\, k = 0, 1, 2, \dots, |S|-1 \rangle$\;
  $F_0 \gets \procName{mkFrontier}(R_{0}, \emptyset)$,
  $\textit{roots} \gets \emptyset$, 
  $\textit{visited} \gets \emptyset$,
  $i \gets 0$\;
  \While{$R_i \neq \emptyset$}{
   $\textit{visited} \gets \textit{visited} \cup F_i$\;
   $R_{i+1} \gets \langle(\parent[v], v) \,:\, v \in F_i \text{ and } \parent[v] \neq v \rangle$ \ncoord[1.5em,1em]{floater}\;
   $\textit{roots} \gets \textit{roots} \cup \{ v \,:\, v\in F_i \text{ where } \parent[v] = v \}$\;
   $F_{i+1} \gets \procName{mkFrontier}(R_{i+1}, \textit{visited})$, 
   $i \gets i + 1$\;
  }
  \tcp{Set up response distribution}
  Create an instance of \textit{RD} with
  $R_\cup = R_0 \oplus R_1 \oplus \dots \oplus R_i$\;

  \tcp{\textbf{Phase II:} Distribute the answers and shorten the paths}
  $D_0 \gets \{ (r,r) \,:\, r \in \textit{roots} \}$,
  $i \gets 0$\;
  \While {$D_i \neq \emptyset$}{
    For each $(v, r) \in D_i$, in parallel, $\parent[v] \gets r$\;
    $D_{i+1} \gets \bigcup_{(v,r) \in D_i} \left\{ (u, r) \,:\, u \in
      \textit{RD}.\procName{allFrom}(v) \text{ and } u \neq \nul \right\}$.
    That is, create $D_{i+1}$ by expanding every $(v, r) \in D_i$ as the entries
    of $\textit{RD}.\procName{allFrom}(v)$ excluding $\nul$, each inheriting
    $r$.\;
    $i \gets i + 1$ 
  }
  For $i=0,1,2\dots,|S|-1$, in parallel, make $\textit{res}[i] \gets \parent[S[i]]$\;
  \Return {$\textit{res}$}
\tikz[overlay,remember picture,auto] {
  \node [mybox, below, right] at (floater) {%
    \begin{minipage}{3in}
      \textbf{def} \procName{mkFrontier}$(R, \textit{visited})$:
      \begin{enumerate}[leftmargin=1em,topsep=0pt,label={\footnotesize\textsf{\arabic*:}},itemsep=0pt,parsep=1pt]
      \item[] // nodes to go to next
      \item $\textit{req} \gets \langle v \,:\, (v, \string_) \in R \land \textbf{ not } \textit{visited}[v] \rangle$ 
      \item \textbf{return} $\procName{removeDup}(\textit{req})$
      \end{enumerate}
    \end{minipage}
  };
}\;
\end{algorithm*}



Whereas the best sequential data structures (e.g.,~\cite{T75}) require
$O((m+q)\alpha(m+q, n))$ work to process $m$ edges and $q$ queries, our basic
data structure from the previous section needs up to $O((m+q)\log n)$ work for
the same input stream.  This section describes improvements that make it match
the best sequential work bound while preserving the polylogarithmic depth
guarantee.  The main result for this section is as follows:

\begin{theorem}
\label{thm:workeff-algo}
There is a bulk-parallel data structure for incremental connectivity over an
infinite window with the following properties:
\begin{enumerate}[label=(\arabic*),topsep=0pt,itemsep=0pt]
\item The total memory consumption is $O(n)$ words.
\item The depth of \procName{Bulk-Update} and \procName{Bulk-Query} is $O(\log n)$ each.
\item Over the lifetime of the data structure, the total work for processing $m$
  edge updates (across all \procName{Bulk-Update}) and $q$ queries is
  $O((m+q)\alpha(m+q, n))$.
\end{enumerate}
\end{theorem}

\myparagraph{Overview:} All sequential data structures with a
$O((m+q)\alpha(n))$ bound use a technique called path compression, which
shortens the path that \find{} traverses on to reach the root, making subsequent
operations cheaper.  Our goal in this section is to enable path compression
during parallel execution.  We present a new parallel \find{} procedure called
\procName{Bulk-Find}, which answers a set of \find{} queries in parallel and
performs path compression.


To understand the benefits of path compression, consider a concrete example in
Figure~\annref[A]{fig:path-compress-illus}, which shows a union-find tree $T$
that is a typical in a union-find forest. The root of $T$ is $r = 19$.  Suppose
we need to support $\find$'s from $u = 1$ and $v = 7$.  When all is done, both
$\find(u)$ and $\find(v)$ should return $r$.  Notice that in this example, the
paths to the root $u \rightsquigarrow r$ and $v \rightsquigarrow r$ meet at
common vertex $w = 4$.  That is, the two paths are identical from $w$ onward to
$r$.  If \find's were done sequentially, say $\find(u)$ before $\find(v)$, then
$\find(u)$---with path compression---would update all nodes on the
$u \rightsquigarrow r$ path to point to $r$.  This means that when $\find(v)$
traverses the tree, the path to the root is significantly shorter: for
$\find(v)$, the next hop after $w$ is already $r$.


The kind of sharing and shortcutting illustrated, however, is not possible when
the \find{} operations are run independently in parallel.  Each \find{}, unaware
of the others, will proceed all the way to the root, missing out on possible
sharing.

We fix this problem by organizing the parallel computation so that the work on
different ``flows'' of \find{}s is carefully coordinated. 
Algorithm~\ref{algo:bulk-find} shows an algorithm \procName{Bulk-Find}, 
which works in two phases, separating actions that only
read from the tree from actions that only write to it:
\begin{description}[leftmargin=0pt]
\item[$\rhd$]\emph{Phase I:} Find the roots for all queries, coalescing flows as
  soon as they meet up.  This phase should be thought of as running
  breadth-first search (BFS), starting from all the query nodes $S$ at once.  As
  with normal BFS, if multiple flows meet up, only one will move on.  Also, if a
  flow encounters a node that has been traversed before, that flow no longer
  needs to go on.  To proceed to Phase II, we need to record the paths traversed
  so that we can distribute responses to the requesting nodes.

\item[$\rhd$]\emph{Phase II:} Distribute the answers and shorten the paths.
  Using the transcript from Phase I, Phase II makes sure that all nodes
  traversed will point to the corresponding root---and answers delivered to all
  the \find{}s.  This phase, too, should be thought of as running breadth-first
  search (BFS) backwards from all the roots reached in Phase I.  This BFS
  reverses the steps taken in Phase I using the trails recorded.  There is a
  technical challenge in implementing this.  Back in Phase I, to minimize the
  cost of recording these trails, the trails are kept as a list of directed
  edges (marked by their two endpoints) traversed.  However, for the reverse
  traversal in Phase II to be efficient, it needs a means to quickly look up all
  the neighbors of a vertex (i.e., at every node, we must be able to find every
  flow that arrived at this node back in Phase I). For this, we design a data
  structure that takes advantage of hashing and integer sorting
  (Theorem~\ref{thm:intsort}) to keep the parallel complexity low.  We discuss
  our solution to this problem in the section that follows
  (Lemma~\ref{lem:resp-dist}).
\end{description}

\myparagraph{Example:} We illustrate how the \procName{Bulk-Find} algorithm
works using the union-find from Figure~\annref[A]{fig:path-compress-illus}.  The
queries to the \procName{Bulk-Find} are nodes that are circled.  The paths
traversed in Phase I are shown in panel B.  If a flow is terminated, the last
edge traversed on that flow is rendered as
\begin{tikzpicture}
\draw[-|] (0,-0.25) -- (0.5,-0.25);
\end{tikzpicture}.%

Notice that as soon as flows meet up, only one of them will carry on.  In
general, if multiple flows meet up at a point, only one will go on.  Notice also
that both the flow $1 \to 2 \to 4$ and the flow $7 \to 8 \to 9 \to 4$ are
stopped at $4$ because $4$ is a source itself, which was started at the same
time as $1$ and $7$.  At the finish of Phase I, the graph (in fact a tree) given
by $R_\cup$ is shown in panel C.  Finally, in Phase II, this graph is traversed
and all nodes visited are updated to point to their corresponding root (as shown
in panel D).






\begin{figure*}[tbh]
  \includegraphics[width=\linewidth]{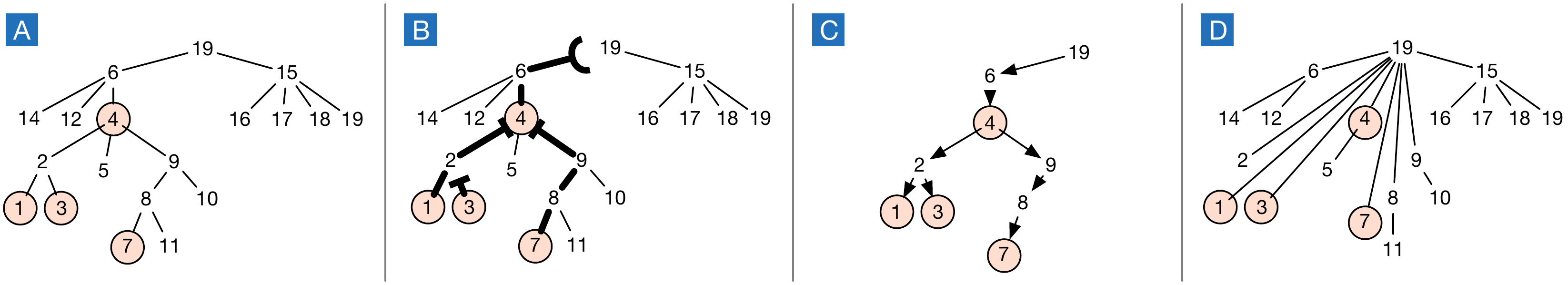}
  \caption{\textsf{A}: An example union-find tree with sample queries circled;
    \textsf{B}: Bolded edges are paths, together with their stopping points,
    that result from the traversal in Phase I; \textsf{C}: The traversal graph
    $R_\cup$ recorded as a result of Phase I; and \textsf{D}: The union-find
    tree after Phase II, which updates all traversed nodes to point to their
    roots.}
  \label{fig:path-compress-illus}
\end{figure*}

\subsection{Response Distributor}

Consider a sequence
$R_\cup = \langle (\textit{from}_i, \textit{to}_i)\rangle_{i=1}^\lambda$.
%
We need a data structure $\textit{RD}$ such that after some preprocessing of
$R_\cup$, can efficiently answer the query $\textit{RD}.\procName{allFrom}(f)$
which returns a sequence containing all $\textit{to}_i$ where
$\textit{from}_i = f$.

To meet the overall running time bound, the preprocessing step cannot take more
than $O(\lambda)$ work and $O(\polylog(\lambda))$ depth.  As far as we know, we
cannot afford to generate, say, a sequence of sequences $RD$ where $RD[f]$ is a
sequence containing all $\textit{to}_i$ such that $\textit{from}_i = f$.
Instead, we propose a data structure with the following properties:
\begin{lemma}[Response Distributor]
  \label{lem:resp-dist}
  There is a data structure response distributor (\textit{RD}) that from input
  $R_\cup = \langle (\textit{from}_i, \textit{to}_i)\rangle_{i=1}^\lambda$ can
  be constructed in $O(\lambda)$ work and $O(\polylog(n))$ depth.  Each
  \procName{allFrom} query can be answered in $O(\log \lambda)$
  depth. Furthermore, if $\mathbb{F}$ is the set of unique $\textit{from}_i$
  (i.e., $\mathbb{F} = \{ \textit{from}_i \,:\, i = 1, \dots, \lambda\}$), then
  \[\expct{\sum_{f \in F} \Work(\textit{RD}.\procName{allFrom}(f))} =
  O(\lambda).
  \]
\end{lemma}
\begin{proof}
  Let $h$ be a hash function from the domain of $\textit{from}_i$'s (a subset of
  $[n]$) to $[\rho]$, where $\rho = 3\lambda$.  To construct an \textit{RD}, we
  proceed as follows. Compute the hash for each $\textit{from}_i$ using
  $h(\cdot)$ and sort the ordered pairs $(\textit{from}_i, \textit{to}_i)$ by
  their hash values.  Call this sorted array $A$. After sorting, we know that
  pairs with the same hash value are stored consecutively in $A$. Now create an
  array $o$ of length $\rho+1$ so that $o_i$ marks the beginning of pairs whose
  hash value is $i$.  If none of them hash to $i$, $o_i = o_{i+1}$.  These steps
  can be done using \procName{intSort} and standard techniques in $O(\lambda)$
  work and $O(\polylog(\lambda))$ depth because the hash values range within
  $O(\lambda)$.

  To support $\procName{allFrom}(f)$, we compute $\kappa = h(f)$ and look in $A$
  between $o_{\kappa}$ and $o_{\kappa+1} - 1$, selecting only pairs whose
  $\textit{from}$ matches $f$.  This requires at most
  $O(\log |o_{\kappa+1} - o_{\kappa}|) = O(\log \lambda)$ depth.  The more
  involved question is how much work is needed to support \procName{allFrom}
  over all.  To answer this, consider all the pairs in $R_\cup$ with
  $\textit{from}_i = f$.  Let $n_f$ denote the number of such pairs.  These
  $n_f$ pairs will be gone through by queries looking for $f$ and other entries
  that happen to hash to the same value as $f$ does.  The exact number of times
  these pairs are gone through is
  $\beta_f := \#\{ s \in \mathbb{F} : h(f) = h(s)\}$.
  Hence, across all queries $ f \in \mathbb{F}$, the total work is
  $\sum_{f \in \mathbb{F}} n_f\beta_f$.  But
  $\expct{\beta_f} \leq 1 + \frac{|\mathbb{F}|}{\rho}$, so
  \[
  \sum_{f \in \mathbb{F}} \expct{n_f\beta_f} \leq \left(1 +
    \tfrac{|\mathbb{F}|}{\rho}\right)\sum_{f \in \mathbb{F}} n_f \leq
  (1+\tfrac{\lambda}{3\lambda})\lambda \leq 2\lambda
  \]
  because $|\mathbb{F}| \leq \lambda$ and
  $\sum_{f \in \mathbb{F}} n_f = \lambda$, completing the proof.
\end{proof}

With this lemma, the cost of $\procName{Bulk-Find}$ can be stated as follows.
\begin{lemma}
  $\procName{Bulk-Find}(U, S)$ does $O(|R_\cup|)$ work and has $O(\polylog(n))$
  depth.
\end{lemma}

\begin{proof}
  The $R_i$'s, $F_i$'s, and $D_i$'s can be maintained directly as arrays.  The
  \textit{roots} and \textit{visited} sets can be maintained as as bit flags on
  top of the vertices of $U$ as all we need are setting the bits
  (adding/removing elements) and reading their values (membership testing).
  There are two phases in this algorithm. In Phase I, the cost of adding $F_i$
  to \textit{visited} in iteration $i$ is bounded by $|R_{i}|$.  Using standard
  parallel operations~\cite{JaJa:book92}, the work of the other steps is clearly
  bounded by $|R_{i+1}|$, including \procName{mkFrontier} because
  \procName{removeDup} does work linear in the input, which is bounded by
  $|R_{i+1}|$. Thus, the work of Phase I is at most
  $O(\sum_i |R_i|) = O(|R_\cup|)$. In terms of depth, because the union-find
  tree has depth at most $O(\log n)$, the \textbf{while} loop can proceed for at
  most $O(\log n)$ times.  Each iteration involves standard operations with
  depth at most $O(\log^2 n)$, so the depth of Phase I is at most $O(\log^3 n)$.

  In Phase II, the dominant cost comes from expanding $D_i$ into $D_{i+1}$ by
  calling $\textit{RD}.\procName{allFrom}$. By Lemma~\ref{lem:resp-dist}, across
  all iterations, the work caused by $\textit{RD}.\procName{allFrom}$, run on
  each vertex once, is expected $O(|R_\cup|)$, and the depth is
  $O(\polylog(|R_\cup|)) \leq O(\polylog(|R_\cup|))$.  Overall, the algorithm
  requires $O(|R_\cup|)$ work and $O(\polylog(n))$ depth.
\end{proof}

\subsection{\procName{Bulk-Find}'s Cost Equivalence to Serial \find{}}
In analyzing the work bound of the improved data structure, we will show that
what \procName{Bulk-Find} does is equivalent to some sequential execution of the
standard \find{} and requires the same amount of work, up to constants.



To gather intuition, we will manually derive such a sequence for the sample
queries $S = \{1, 3, 4, 7\}$ used in Figure~\ref{fig:path-compress-illus}.  The
query of $4$ went all the way to the root without merging with another flow.
But the queries of $1$ and $7$ were stopped at $4$ and in this sense, depended
upon the response from the query of $4$. By the same reasoning, because the
query of $3$ merged with the query of $1$ (with $1$ proceeding on), the query of
$3$ depended on the response from the query of $1$.  Note that in this view,
although the query of $3$ technically waited for the response at $2$, it was the
query of $1$ that brought the response, so it depended on $1$.  To derive a
sequence execution, we need to respect the ``depended on'' relation: if $a$
depended on $b$, then $a$ will be invoked after $b$.  As an example, one
sequential execution order that respects these dependencies is
$\find(4), \find(7), \find(1), \find(3)$.

We can check that by applying \find{}s in this order, the paths traversed are
exactly what the parallel execution does as $U.\find{}$ performs full path
compression.

We formalize this idea in the following lemma: 
\begin{lemma}
  \label{lem:we-equiv}
  For a sequence of queries $S$ with which $\procName{Bulk-Find}(U, S)$ is
  invoked, there is a sequence $S'$ that is a permutation $S$ such that applying
  $U.\find$ to $S'$ serially in that order yields the same union-find forest as
  \procName{Bulk-Find}'s and incurs the same traversal cost of $O(|R_\cup|)$,
  where $R_\cup$ is as defined in the \procName{Bulk-Find} algorithm.
\end{lemma}

\begin{proof}
  For this analysis, we will associate every
  $(\textit{parent}, \textit{child}) \in R_\cup$ with a query $q \in S$.
  Logically, every query $q \in S$ starts a flow at $q$ ascending up the
  tree. If there are multiple flows reaching the same node, \procName{removeDup}
  inside \procName{mkFrontier} decides which flow to go on.  From this view, for
  any nonroot node $u$ appearing in $R_\cup$, there is \emph{exactly} one query
  flow from this node that proceeds up the tree.  We will denote this flow by
  $\textsf{own}(u)$.
  
  If a query flow is stopped partway (without reaching the corresponding root),
  the reason is either it merges in with another flow (via
  \procName{mkFrontier}) or it recognizes another flow that visited where it is
  going before (via \textit{visited}). For every query $q$ that is stopped partway,
  let $r(q)$ be the furthest point in the tree it has advanced to, i.e., $r(q)$
  is the endpoint of the maximal path in $R_\cup$ for the query flow $q$.

  \newcommand{\depg}[0]{\ensuremath{G_{\textit{dep}}}}
  In this set up, a query flow whose furthest point is $u$ will depend on the
  response from the query $\textsf{own}(u)$.  Therefore, we form a dependency
  graph $G_{\textit{dep}}$ (``$u$ depends on $v$'') as follows. The vertices are
  all the vertices from $S$.  For every query flow $q$ that is stopped partway,
  there is an arc $\textsf{own}(r(q)) \to q$.

  Let $S'$ be a topologically-ordered sequence of $\depg$. Multiple copies of
  the same query vertex can simply be placed next to each other.  If we apply
  $U.\find{}$ serially on $S'$, then all queries that a query vertex $q$ depends
  on in $\depg$ will have been called prior to $U.\find(q)$. Because of full
  path compression, this means that $U.\find(q)$ will follow
  $u \rightsquigarrow r(q) \to t$ ($r(q) \to t$ is one step), where $t$ is the
  root of the tree.  Hence, every $U.\find(q)$ traverses the same number of
  edges as $u \rightsquigarrow r(q)$ plus $1$. As every $R_\cup$ edge is part of
  a query flow, we conclude that the work of running $U.\find$ on $S'$ in that
  order is $O(|R_\cup|)$.
  \let\depg\undefined
\end{proof}

Finally, to obtain the bounds in Theorem~\ref{thm:workeff-algo}, we modify
\procName{Simple-Bulk-Query} and \procName{Simple-Bulk-Update} (in the
relabeling step) to use \procName{Bulk-Find} on all query pairs.  The depth
clearly remains $O(\polylog(n))$ per bulk operation. Aggregating the cost of
\procName{Bulk-Find} across calls from \procName{Bulk-Update} and
\procName{Bulk-Query}, we know from Lemma~\ref{lem:we-equiv} that there is a
sequential order that has the same work.  Therefore, the total work is bounded
by $O((m+q)\alpha(m+q, n))$.

\section{Implementation and Evaluation}
\label{sec:eval}

This section discusses an implementation of the proposed data structure and its
empirical performance.

\subsection{Implementation} 
With an eye toward a simple implementation that delivers good practical
performance, we set out to implement the simple bulk-parallel data structure
from Section~\ref{sec:simple-algo}.  The underlying union-find data structure
$U$ maintains two arrays of length $n$---\parent{}~and~\procName{sizes}---one
storing a parent pointer for each vertex, and the other tracking the sizes of
the trees.  The \find{} and \union{} operations follow a standard textbook
implementation. On top of these operations, we implemented
\procName{Simple-Bulk-Query} and \procName{Simple-Bulk-Update} as described
earlier in the paper.  We use standard sequence manipulation operations
(e.g., filter, prefix sum, pack, remove duplicate) from the PBBS
library~\cite{ShunBFGKST:spaa12}. There are two modifications that we made to
improve practical performance of the implementation:
\begin{description}[leftmargin=0pt,parsep=0pt]
\item[]\emph{Path Compression:} We wanted some benefits of path compression but without
  the full complexity of the work-efficient parallel algorithm from Section~\ref{sec:work-efficient}, 
  to keep the code simple.  We settled with the following pragmatic solution: The \find{}
  operations inside \procName{Simple-Bulk-Query} and \procName{Simple-Bulk-Update} still
  run independently in parallel.  But after finding the root, each operation
  traverses the tree one more time to update all the nodes on the path to point
  to the root. This leads to shorter paths for later bulk operations with clear
  performance benefits.  However, for large bulk sizes, the approach may still perform
  significantly more work than the work-efficient solution because the path compression
  from a \find{} operation may not benefit other \find{} operations within the same minibatch.
\item[]\emph{Connected Components:} The algorithm as described uses as a
  subroutine a linear-work parallel algorithm to find connected
  components. These linear work algorithms expect a graph representation that
  gives quick random access to the neighbors of a vertex.  We found the
  processing cost to meet this requirement to be very high and instead
  implemented the algorithm for connectivity described in
  Blelloch~et~al.~\cite{BlellochFGS:ppopp12}. Although this has worse
  theoretical guarantees, it can work with a sequence of edges directly and
  delivers good real-world performance.
\end{description}


\subsection{Experimental Setup}
\myparagraph{Environment:}
We performed experiments on an Amazon EC2 instance with $20$ cores (allowing for
$40$ threads via hyperthreading) of $2.4$ GHz Intel Xeon E5-2676 v3 processors,
running Linux 3.11.0-19 (Ubuntu 14.04.3).  We believe this represents a baseline
configuration of midrange workstations available in a modern cluster.
All programs were compiled with Clang version 3.4 using the flag \texttt{-O3}.
This version of Clang has the Intel Cilk runtime, which implements a
work-stealing scheduler known to impose only a small overhead on both parallel
and sequential code.
We report wall-clock time measured using
\verb@std::chrono::high_resolution_clock@.

For robustness, we perform \emph{three} trials and report the median running
time. Although there is randomness involved in the connected component (CC)
algorithm, we found no significant fluctuations in the running time across
runs.

\myparagraph{Datasets:} Our study aims to study the behavior of the algorithm on
a variety of graph streams.  To this end, we use a collection of synthetic graph
streams created using well-accepted generators.  We include both power-law-type
graphs and more regular graphs in the experiments.  These are graphs commonly
used in dynamic/streaming graph experiments~(e.g., \cite{MGB13}). A summary of
these datasets appear in
Table~\ref{tbl:datastat}.  

\begin{table}[t]
\centering
\begin{tabular}{l r r p{1.75in}}
  \toprule
  Graph & \#Vertices & \#Edges & Notes\\
  \midrule
  3Dgrid  &  $99.9$M & $300$M & {\footnotesize 3-d mesh}\\ 
  random  &  $100$M& $500$M & {\footnotesize $5$ randomly-chosen neighbors per node} \\
  local5  &  $100$M& $500$M & {\footnotesize small separators, avg. degree $5$}\\
  local16 &  $100$M& $1.6$B & {\footnotesize small separators, avg. degree $16$}\\
  rMat5   & $134$M& $500$M & {\footnotesize power-law graph using rMat~\cite{ChakrabartiZF:sdm04}}\\
  rMat16  & $134$M& $1.6$B & {\footnotesize a denser rMat graph}\\
  \bottomrule
\end{tabular}
\caption{Characteristics of the graph streams used in our experiments,
  showing for every dataset, the total number of nodes ($n$), 
  the total number of edges ($m$), and a brief description.}
\label{tbl:datastat}
\end{table}


The graph streams in our experiments differ substantially in how quickly they
become connected.  This input characteristic influences the data structure's
performance.  In Figure~\ref{fig:numconn}, we show for each graph stream, the
number of connected components at different points in the stream. The local16
graph becomes fully connected right around the midpoint of the stream.  Both
rMat5 and rMat16 continue to have tens of millions of components after consuming
the whole stream.  Note that in this figure, random and local5 are almost
visually indistinguishable until the very end.

\begin{figure}
  \centering
  \includegraphics[width=0.5\linewidth]{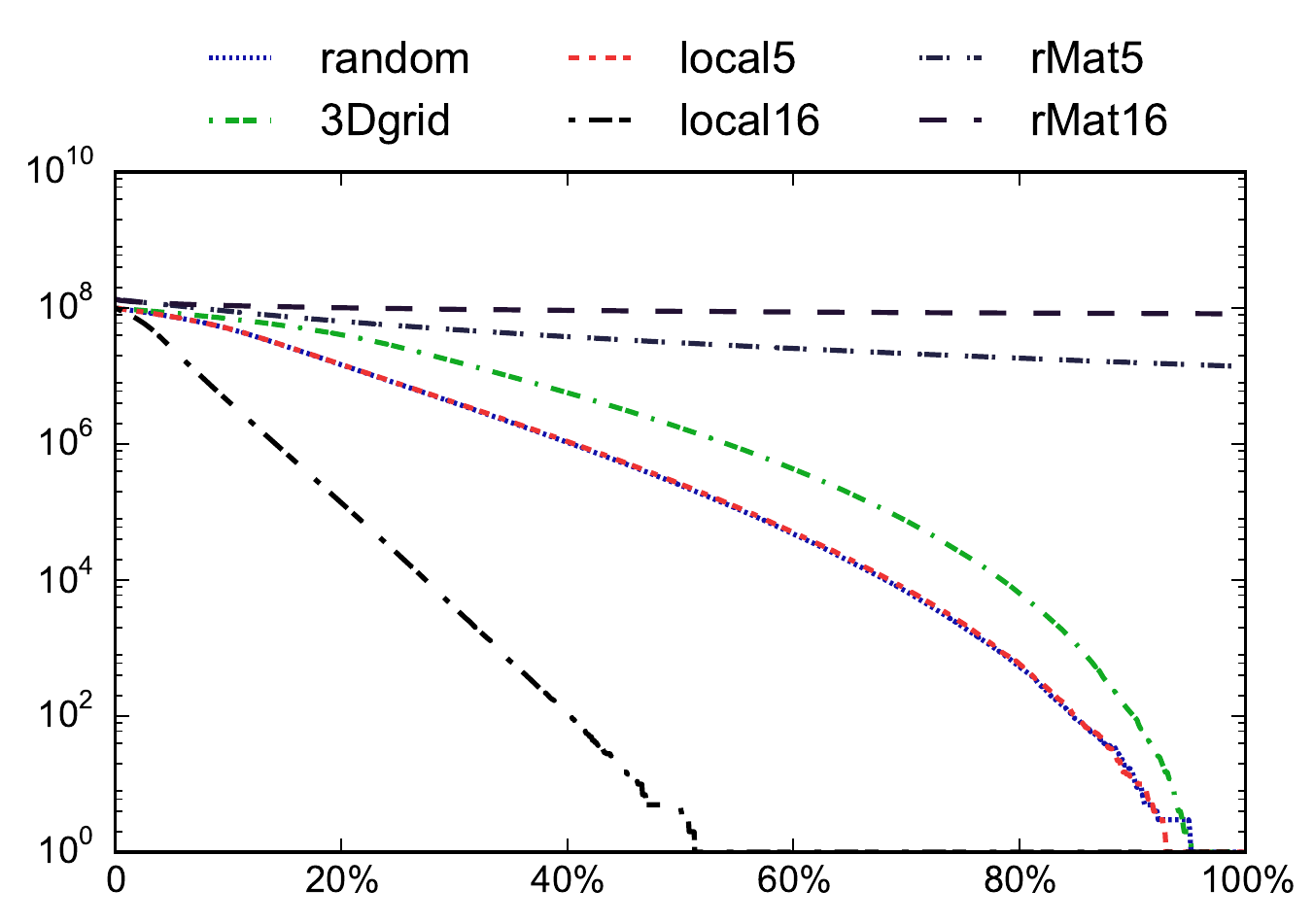}
  \caption{The numbers of connected components for each graph dataset at
    different percentages of the total graph stream processed.}
  \label{fig:numconn}
\end{figure}

\myparagraph{Baseline:} We directly compare our algorithms with union find
(denoted \procName{UF}), using both the union by size and a path compression
variant, which has the optimal sequential running time.  Most prior algorithms
either focus on parallel graphs or streaming graphs, not parallel streaming
graphs.  We note that the algorithm of McColl et~al.~\cite{MGB13} that works in
the parallel dynamic setting is not directly comparable to ours. Their algorithm
focuses on supporting insertion and deletion of arbitrary edges, whereas ours is
designed to take advantage of the insert-only setting.


\subsection{Results}
\label{sec:exp-results}

\emph{How does the bulk-parallel data structure perform on a multicore machine?}
To this end, we investigate the parallel overhead, speedup, and scalability
trend.

Table~\ref{table:overhead} shows the timings for the baseline sequential
implementation of union-find \procName{UF} with and without path compression and
the bulk-parallel implementation {\em on a single thread} for four different
batch sizes, 500K, 1M, 5M, and 10M.  To measure overhead, we first compare our
implementation to union find \emph{without} path compression: our implementation
is between $1.01$x and $2.5$x slower except on local16, in which the bulk
parallel achieves some speedups even on one thread.  This is mainly because the
number of connected components in local16 drops quickly to 1 as soon as
midstream (Figure~\ref{fig:numconn}).  With only 1 connected component, there is
little work for bulk-parallel to be done after that.  Compared to union find
with path compression, our implementation, which does pragmatic path
compression, shows nontrivial---but still acceptable---overhead, as to be
expected because our solution does not fully benefit \find{}s within the same
minibatch.



\begin{table}[t]
\centering
\begin{tabular}{lrrrrrr}
\toprule
  \multirow{2}{*}{\emph{Graph}}
  &\procName{UF}& \procName{UF}& \multicolumn{4}{c}{Bulk-Parallel Using Batch Size}\\
\cmidrule{4-7}
  &(no p.c.) & (p.c.) & 500K & 1M & 5M & 10M \\
\midrule
random & 44.63 & 18.42 & 65.43 & 66.57 & 75.20 & 77.89\\
3Dgrid & 30.26 & 14.37 & 61.10 & 62.00 & 71.74 & 75.07\\
local5 & 44.94 & 18.51 & 65.84 & 66.77 & 75.33 & 78.23\\
local16 & 154.40 & 46.12 & 114.34 & 108.92 & 114.80 & 117.55\\
rMat5 & 33.39 & 18.47 & 66.98 & 68.48 & 74.97 & 78.69\\
rMat16 & 81.74 & 35.29 & 83.27 & 76.64 & 76.03 & 77.62\\

\bottomrule
\end{tabular}
\caption{Running times (\textbf{in seconds}) on $1$ thread of the baseline union-find 
  implementation \procName{UF} with and without path compression  (unaffected by the batch size) and the bulk-parallel data structure as the batch size is varied.}
\label{table:overhead}
\end{table}

Table~\ref{table:speedup} shows the average throughputs (million edges/second)
of \procName{Bulk-Update} for different batch sizes. Here $T_1$ denotes the
throughput on $1$ thread and $T_{20c}$ the throughput on $20$ cores (40
hyper-threads). We also show the speedup as measured by $T_{20c}/T_1$.  We
observe consistent speedup on all six datasets under all four batch sizes.
Across all datasets, the general trend is that the larger the batch size, the
higher was the speedup. This is to be expected, since a larger batch size means
more work per core in processing each batch, and lesser overhead of
synchronization.

\begin{center}
\begin{table}[t]
{\small
\begin{tabular}{lrrrrrrrrrrrr}
  \toprule
  \multirow{2}{*}{\emph{Graph}}
  &\multicolumn{3}{c}{Using $b=500$K}
  &\multicolumn{3}{c}{Using $b=1$M}
  &\multicolumn{3}{c}{Using $b=5$M}
  &\multicolumn{3}{c}{Using $b=10$M}\\
  \cmidrule(r){2-4}
  \cmidrule(r){5-7}
  \cmidrule(r){8-10}
  \cmidrule(r){11-13}
  & $T_1$ & $T_{20c}$ & $\sfrac{T_{20c}}{T_1}$ 
  & $T_1$ & $T_{20c}$ & $\sfrac{T_{20c}}{T_1}$ 
  & $T_1$ & $T_{20c}$ & $\sfrac{T_{20c}}{T_1}$ 
  & $T_1$ & $T_{20c}$ & $\sfrac{T_{20c}}{T_1}$ \\
  \midrule
  random & 7.64 & 36.87 & $4.8$x & 7.51 & 46.02 & $6.1$x & 6.65 & 60.66 & $9.1$x & 6.42 & 63.90 & $10.0$x\\
3Dgrid & 4.91 & 27.97 & $5.7$x & 4.83 & 34.97 & $7.2$x & 4.18 & 44.27 & $10.6$x & 3.99 & 45.24 & $11.3$x\\
local5 & 7.59 & 38.41 & $5.1$x & 7.49 & 48.32 & $6.5$x & 6.64 & 64.61 & $9.7$x & 6.39 & 64.09 & $10.0$x\\
local16 & 13.99 & 78.83 & $5.6$x & 14.69 & 95.57 & $6.5$x & 13.94 & 122.69 & $8.8$x & 13.61 & 122.03 & $9.0$x\\
rMat5 & 7.47 & 26.08 & $3.5$x & 7.30 & 34.19 & $4.7$x & 6.67 & 49.92 & $7.5$x & 6.35 & 50.37 & $7.9$x\\
rMat16 & 19.21 & 54.94 & $2.9$x & 20.88 & 78.10 & $3.7$x & 21.05 & 143.63 & $6.8$x & 20.61 & 167.68 & $8.1$x\\

  \bottomrule
\end{tabular}
\caption{Average throughput (\textbf{in million edges/second}) and speedup of \procName{Bulk-Update} for different batch sizes $b$, where $T_1$ is throughput on $1$ thread and $T_{20c}$ is the throughput on $20$ cores.
}
\label{table:speedup}
}
\end{table}
\end{center}


Figure~\ref{fig:scalability} shows the average throughput (edges/sec) as the
number of threads increases from 1 to $20c$, which represents 40 hyperthreads.
Three different batch sizes were used for the experiments: 1M, 5M and 10M.  The
top chart represents the results on the random dataset, the middle chart on the
local16 dataset and the bottom chart on the rMat16 dataset.  In general, as the
number of threads increases, the average throughput increases for all 3 datasets
under different batch sizes. With a 10M batch size on 20 cores, we observe
speedups between $8$--$11$x. On the rMat16 dataset (the bottom chart), the
throughput starts to drop with batch size of 1M when the number of threads
increases beyond 20.  This is due to the relatively large number of connected
components in the rMat16 dataset from the beginning towards the end of
processing the entire dataset(see Figure~\ref{fig:numconn}).  In this case, the
work done per batch of input edges is relatively small, and a 1M batch size is
too small for the rMat16 dataset to realize additional parallelization benefits
beyond 20 threads.


\begin{figure}
\centering
\includegraphics[width=0.5\textwidth]{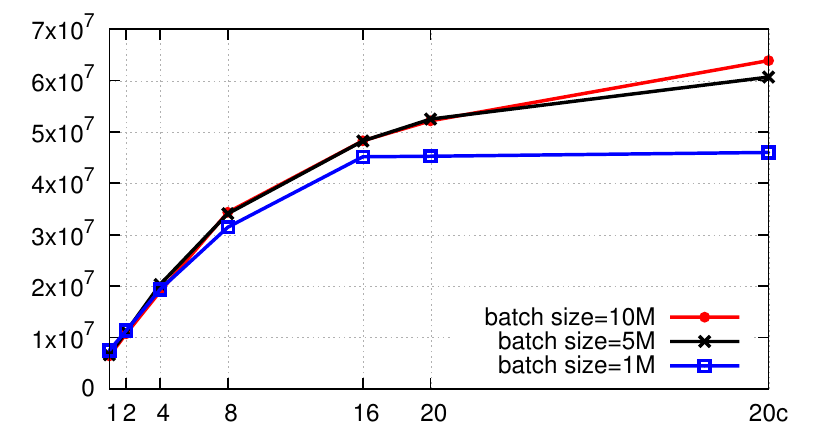}
\includegraphics[width=0.5\textwidth]{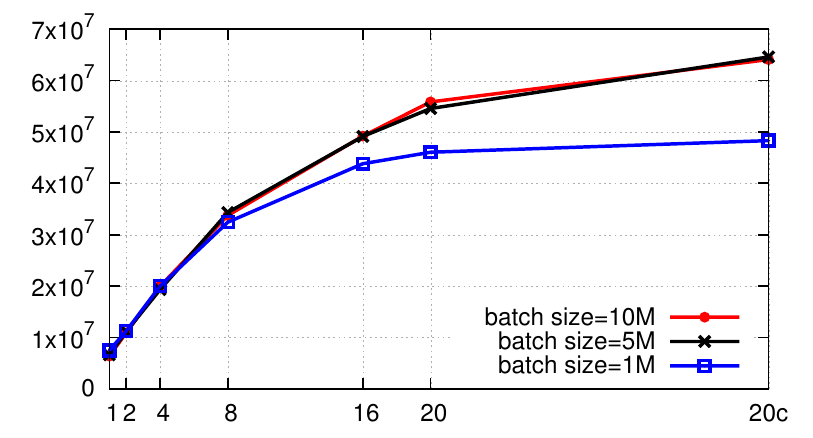}
\includegraphics[width=0.5\textwidth]{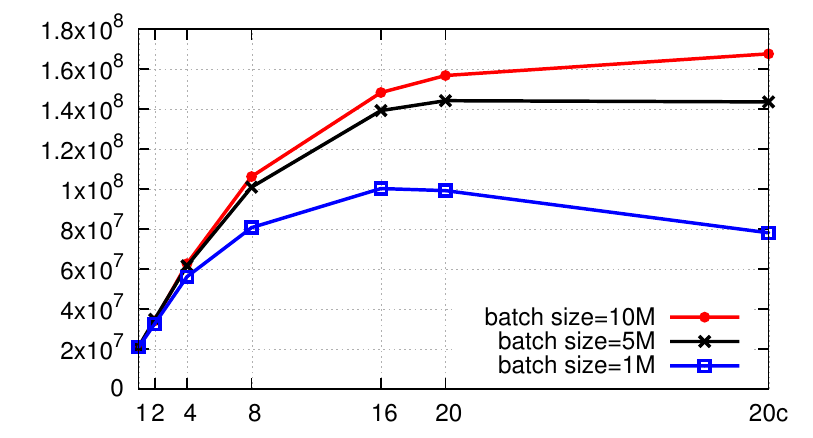}
\caption{Average throughput (\textbf{edges per second}) as the number of threads is
  varied from 1 to $40$ (denoted by $20c$ as they run on 20
  cores with hyperthreading).  The graph streams shown are (top) random,
  (middle) local16, and (bottom) rMat16.}
\label{fig:scalability}
\end{figure}


\section{Conclusion}
\label{sec:concl}

We presented a shared-memory parallel algorithm for incremental graph connectivity in the minibatch arrival model.  Our algorithm has polylogarithmic parallel depth and its total work across all processors is of the same order as the work due to the best sequential algorithm for incremental graph connectivity. We also presented a simpler parallel algorithm that is easier to implement and has good practical performance.

This presents several natural open research questions. We list some of them here. (1)~In case all edge updates are in a single minibatch, the total work of our algorithm is (in a theoretical sense), superlinear in the number of edges in the graph. Whereas, the optimal batch algorithm for graph connectivity, based on a depth-first search, has work linear in the number of edges. Is it possible to have an incremental algorithm whose work is linear in the case of very large batches, such as the above, and falls back to the union-find type algorithms for smaller minibatches? Note that for all practical purposes, the work of our algorithm is linear in the number of edges, due to very slow growth of the inverse Ackerman's function. (2)~Can these results on parallel algorithms be extended to the fully dynamic case when there are both edge arrivals as well as deletions?





\end{document}